\documentclass[reqno, eucal]{amsart}

\usepackage[hmargin={3cm,3cm},vmargin={3cm,3cm},includehead]{geometry}
\usepackage{amsmath,amssymb}
\usepackage{amsfonts}
\usepackage{mathrsfs}
\newcommand{\ds}{\displaystyle}

\def\EXP{\textrm{{\large e}}}

\newcommand{\bos}{{\bf a}}

\newcommand{\kop}{{\bf k}}

\newcommand{\pop}{\boldsymbol{p}}
\newcommand{\sop}{\boldsymbol{\sigma}}

\newcommand{\wop}{{\bf w}}

\newcommand{\ii}{{\mathrm i}}

\usepackage[mathscr]{eucal} %  use \EuScript (\mathcal unchanged)

\numberwithin{equation}{section}

\usepackage[dvips]{graphicx}
\usepackage[dvips]{color}
\usepackage{amsmath,amsfonts,amssymb,amsthm,amscd}
\usepackage{epic}
\usepackage{eepic}
\usepackage{longtable}
\usepackage{array}

\newtheorem{theorem}{Theorem}[section]

\newtheorem{proposition}[theorem]{Proposition}

\theoremstyle{definition}

\newtheorem{remark}[theorem]{Remark}

\newcommand{\Z}{{\mathbb Z}}
\newcommand{\R}{{\mathbb R}}
\newcommand{\C}{{\mathbb C}}
\newcommand{\A}{{\mathcal A}}
\newcommand{\I}{{\mathrm i}}

\begin{document}

\title[Tetrahedron Equation
and modular double]{Tetrahedron Equation
and Quantum $\boldsymbol{R}$ Matrices
for modular double of 
$\boldsymbol{U_q(D^{(2)}_{n+1})}, 
\boldsymbol{U_q(A^{(2)}_{2n})}$ and  
$\boldsymbol{U_q(C^{(1)}_{n})}$}

\author{Atsuo Kuniba}
\email{atsuo@gokutan.c.u-tokyo.ac.jp}
\address{Institute of Physics, University of Tokyo, Komaba, Tokyo 153-8902, Japan}

\author{Masato Okado}
\email{okado@sci.osaka-cu.ac.jp}
\address{Department of Mathematics, Osaka City University, 
3-3-138, Sugimoto, Sumiyoshi-ku, Osaka, 558-8585, Japan}

\author{Sergey Sergeev}
\email{Sergey.Sergeev@canberra.edu.au}
\address{Faculty of Education, Science, Technology, Engineering and Mathematics, University of Canberra, ACT 2106 Australia}

\maketitle

%\vspace{0.5cm}
\begin{center}{\bf Abstract}
\end{center}
\vspace{0.4cm}
We introduce a homomorphism from the quantum affine algebras 
$U_q(D^{(2)}_{n+1}), U_q(A^{(2)}_{2n}), U_q(C^{(1)}_{n})$
to the $n$-fold tensor product of the $q$-oscillator algebra ${\mathcal A}_q$.
Their action commutes with the solutions of the Yang-Baxter equation
obtained by reducing the solutions of the tetrahedron equation
associated with the modular and the Fock representations of $\A_q$.
In the former case,
the commutativity is enhanced to 
the modular double of these quantum affine algebras.

\vspace{0.5cm}\noindent


\section{Introduction}\label{sec:intro}

The tetrahedron equation \cite{Zam80} is a three dimensional (3D)
generalization of the Yang-Baxter equation \cite{Bax}.
Among its several versions, the basic one adapted to homogeneous 
3D vertex models has the form
\begin{align}\label{te}
R_{1,2,4}R_{1,3,5}R_{2,3,6}R_{4,5,6} = R_{4,5,6}R_{2,3,6}R_{1,3,5}R_{1,2,4},
\end{align}
where $R$ is a linear operator on the tensor cube of some vector space.
The equality holds for the operators on its sixfold tensor product, where 
the indices specify the components on which $R$ acts nontrivially. 
We call a solution to the tetrahedron equation a 3D $R$.

Several solutions have been found until now with some important 
clues to the relevant algebraic structures
such as the quantized coordinate ring of $\mathrm{SL}_3$ \cite{KV}, 
PBW basis of the nilpotent subalgebra of $U_q(sl_3)$ \cite{S},
the $q$-oscillator algebra $\A_q$ \cite{BMS,BS} and so forth.
It is known \cite{KO1}  that 
the 3D $R$  associated with 
the Fock representation of $\A_q$ \cite{BS} coincides with 
the one in \cite{KV}.

The tetrahedron equation reduces to the Yang Baxter equation
\[
R_{1,2} R_{1,3} R_{2,3}  = R_{2,3}R_{1,3} R_{1,2}
\]
if the spaces $4,5,6$ are evaluated away suitably \cite{S2,KasV}.
In \cite{KS, KO4}, such a reduction was achieved by taking the
matrix elements with respect to certain boundary vectors.
The resulting solutions to the Yang-Baxter equation were
identified with the quantum $R$ matrices of 
various quantum affine algebras $U_q(\mathfrak{g})$ 
and their representations. 

In this paper, we exploit further aspects 
of the 3D $R$ associated with what we call the modular and the Fock 
representations of the $q$-oscillator algebra $\A_q$.
We use the two boundary vectors to generate
four families of solutions to the Yang-Baxter equation for 
each representation.
Our first result, Theorem \ref{th:main1},  is that they commute with 
the quantum affine algebras 
$U_q(\mathfrak{g})$ with 
$\mathfrak{g} = D^{(2)}_{n+1}, C^{(1)}_n, A^{(2)}_{2n}$ and 
$\tilde{A}^{(2)}_{2n}$ (see Section \ref{ss:qaa} for the definition).
The essential ingredient for this statement is 
a new homomorphism from $U_q(\mathfrak{g})$ to $\A_q^{\otimes n}$ 
in Proposition \ref{pr:pi}.
The two boundary vectors correspond to 
the short and the long simple roots of $\mathfrak{g}$ 
at the two ends of the Dynkin diagram.

Our second result, Theorem \ref{th:main2}, 
is obtained by applying Theorem \ref{th:main1} 
to the modular representation of the pair $(\A_q,\A_{\tilde q})$ such that 
$(\log q)(\log {\tilde q}) = -\pi^2$.
We find that the symmetry of 
the relevant solutions of the Yang-Baxter equation
is enhanced naturally from $U_q(\mathfrak{g})$ 
to its modular double 
$U_q(\mathfrak{g}) \otimes U_{\tilde q}({}^L\mathfrak{g})$
where  ${}^L\mathfrak{g}$ is the Langlands dual of 
$\mathfrak{g}$.
The key to this result is Proposition \ref{pr:key} showing that 
the two boundary vectors 
interchange their role when passing to the modular dual.
An analogous feature has been observed in \cite{Ip1}.
For general background on modular double, we refer to \cite{F,FI}.

The layout of the paper is as follows.
In Section \ref{sec:Uq}, the algebra homomorphism from
$U_q(\mathfrak{g})$ to $\A_q^{\otimes n}$
is presented in Proposition \ref{pr:pi}. 
In Section \ref{sec:3DR},
the 3D $R$ \cite{BMS,BS} and 
the characterization of the boundary vectors \cite{KS} are recapitulated. 
In Section \ref{sec:ybe},
reduction to the Yang-Baxter equation \cite{KS,KO4} is explained and 
the symmetry of the consequent solution is described in Theorem \ref{th:main1}.
All the arguments until this point are valid either for 
the modular or the Fock representations of $\A_q$.
They systematize the proof of the commutativity significantly.
In Section \ref{sec:mod}, the general construction in the 
preceding sections are embodied in the modular representation.
The boundary vectors in this representation are new and 
described explicitly in terms of their wave functions.
They lead to our main result, Theorem \ref{th:main2}.
In Section \ref{sec:fock}, it is explained how the specialization of the 
general results in Sections \ref{sec:Uq}--\ref{sec:ybe} 
to the Fock representation
covers an essential part of the earlier result \cite{KO4}.
Appendix \ref{app:chi} contains identities involving 
the boundary wave function $\chi_b(\sigma)$ and the 
quantum dilogarithm.

\section{Quantum Affine Algebras and $q$-Oscillator Algebra}\label{sec:Uq}

\subsection{Quantum affine algebras}\label{ss:qaa}
We assume that $q$ is generic except in Section \ref{sec:mod} and 
Appendix \ref{app:chi}.
(The basic parameter is $q^{\frac{1}{2}}$ rather than $q$ 
in our convention.) 
The Drinfeld-Jimbo quantum affine algebras (without derivation operator) 
$U_q=U_q(A^{(2)}_{2n}), U_q({\tilde A}^{(2)}_{2n})$, 
$U_q(C^{(1)}_{n})$ and $U_q(D^{(2)}_{n+1})$
are the Hopf algebras 
generated by $e_i, f_i, k^{\pm 1}_i\, (0 \le i \le n)$ satisfying the relations
\cite{D,Ji}
\begin{equation}\label{uqdef}
\begin{split}
&k_i k^{-1}_i = k^{-1}_i k_i = 1,\quad [k_i, k_j]=0,\\
&k_ie_jk^{-1}_i = q_i^{a_{ij}}e_j,\quad 
k_if_jk^{-1}_i = q_i^{-a_{ij}}f_j,\quad
[e_i, f_j]=\delta_{ij}\frac{k_i-k^{-1}_i}{q_i-q^{-1}_i},\\
&\sum_{\nu=0}^{1-a_{ij}}(-1)^\nu
e^{(1-a_{ij}-\nu)}_i e_j e_i^{(\nu)}=0,
\quad
\sum_{\nu=0}^{1-a_{ij}}(-1)^\nu
f^{(1-a_{ij}-\nu)}_i f_j f_i^{(\nu)}=0\;\;(i\neq j),
\end{split}
\end{equation}
where $e^{(\nu)}_i = e^\nu_i/[\nu]_{q_i}!, \,
f^{(\nu)}_i = f^\nu_i/[\nu]_{q_i}!$
and 
$[m]_q! = \prod_{k=1}^m [k]_q$ with 
$[m]_q = \frac{q^m-q^{-m}}{q-q^{-1}}$.
The Cartan matrix $(a_{ij})_{0 \le i,j \le n}$ \cite{Kac} is given by
\begin{align*}
a_{i,j} = 2\delta_{i,j}-\max((\log q_j)/(\log q_i),1)\delta_{|i-j|,1}.
\end{align*}
The data $q_i$ are specified 
above the corresponding vertex $i\, (0 \le i \le n)$ 
in the Dynkin diagrams:
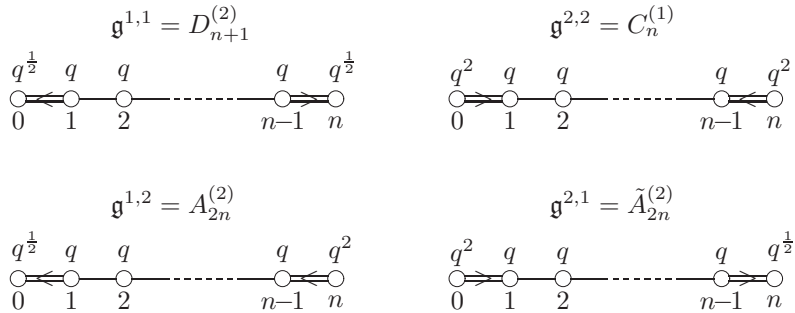
\begin{align*}
\begin{picture}(126,60)(0,-15)
\put(-20,0){
\put(35,25){$\mathfrak{g}^{1,1}=D^{(2)}_{n+1}$}
\multiput( 0,0)(20,0){3}{\circle{6}}
\multiput(100,0)(20,0){2}{\circle{6}}
\multiput(23,0)(20,0){2}{\line(1,0){14}}
\put(83,0){\line(1,0){14}}
\multiput( 2.85,-1)(0,2){2}{\line(1,0){14.3}} %double line
\multiput(102.85,-1)(0,2){2}{\line(1,0){14.3}} %double line
\multiput(59,0)(4,0){6}{\line(1,0){2}} %dash line
\put(10,0){\makebox(0,0){$<$}}
\put(110,0){\makebox(0,0){$>$}}
\put(0,-5){\makebox(0,0)[t]{$0$}}
\put(20,-5){\makebox(0,0)[t]{$1$}}
\put(40,-5){\makebox(0,0)[t]{$2$}}
\put(100,-5){\makebox(0,0)[t]{$n\!\! -\!\! 1$}}
\put(120,-6.5){\makebox(0,0)[t]{$n$}}
\put(3,18){\makebox(0,0)[t]{$q^{\frac{1}{2}}$}}
\put(20,13){\makebox(0,0)[t]{$q$}}
\put(40,13){\makebox(0,0)[t]{$q$}}
\put(100,13){\makebox(0,0)[t]{$q$}}
\put(123,18){\makebox(0,0)[t]{$q^{\frac{1}{2}}$}}
}
\end{picture}
\begin{picture}(126,60)(0,-15)
\put(20,0){
\put(35,25){$\mathfrak{g}^{2,2}=C^{(1)}_{n}$}
\multiput( 0,0)(20,0){3}{\circle{6}}
\multiput(100,0)(20,0){2}{\circle{6}}
\multiput(23,0)(20,0){2}{\line(1,0){14}}
\put(83,0){\line(1,0){14}}
\multiput( 2.85,-1)(0,2){2}{\line(1,0){14.3}} %double line
\multiput(102.85,-1)(0,2){2}{\line(1,0){14.3}} %double line
\multiput(59,0)(4,0){6}{\line(1,0){2}} %dash line
\put(10,0){\makebox(0,0){$>$}}
\put(110,0){\makebox(0,0){$<$}}
\put(0,-5){\makebox(0,0)[t]{$0$}}
\put(20,-5){\makebox(0,0)[t]{$1$}}
\put(40,-5){\makebox(0,0)[t]{$2$}}
\put(100,-5){\makebox(0,0)[t]{$n\!\! -\!\! 1$}}
\put(120,-6.5){\makebox(0,0)[t]{$n$}}
\put(2,15){\makebox(0,0)[t]{$q^2$}}
\put(20,12){\makebox(0,0)[t]{$q$}}
\put(40,12){\makebox(0,0)[t]{$q$}}
\put(100,12){\makebox(0,0)[t]{$q$}}
\put(122,15){\makebox(0,0)[t]{$q^2$}}
}
\end{picture}
\\
\begin{picture}(126,60)(0,-15)
\put(-20,0){
\put(35,25){$\mathfrak{g}^{1,2}=A^{(2)}_{2n}$}
\multiput( 0,0)(20,0){3}{\circle{6}}
\multiput(100,0)(20,0){2}{\circle{6}}
\multiput(23,0)(20,0){2}{\line(1,0){14}}
\put(83,0){\line(1,0){14}}
\multiput( 2.85,-1)(0,2){2}{\line(1,0){14.3}} %double line
\multiput(102.85,-1)(0,2){2}{\line(1,0){14.3}} %double line
\multiput(59,0)(4,0){6}{\line(1,0){2}} %dash line
\put(10,0){\makebox(0,0){$<$}}
\put(110,0){\makebox(0,0){$<$}}
\put(0,-5){\makebox(0,0)[t]{$0$}}
\put(20,-5){\makebox(0,0)[t]{$1$}}
\put(40,-5){\makebox(0,0)[t]{$2$}}
\put(100,-5){\makebox(0,0)[t]{$n\!\! -\!\! 1$}}
\put(120,-6.5){\makebox(0,0)[t]{$n$}}
\put(3,18){\makebox(0,0)[t]{$q^{\frac{1}{2}}$}}
\put(20,12.5){\makebox(0,0)[t]{$q$}}
\put(40,12.5){\makebox(0,0)[t]{$q$}}
\put(100,12.5){\makebox(0,0)[t]{$q$}}
\put(122,16){\makebox(0,0)[t]{$q^2$}}
}
\end{picture}
\begin{picture}(126,60)(0,-15)
\put(20,0){
\put(35,25){$\mathfrak{g}^{2,1}={\tilde A}^{(2)}_{2n}$}
\multiput( 0,0)(20,0){3}{\circle{6}}
\multiput(100,0)(20,0){2}{\circle{6}}
\multiput(23,0)(20,0){2}{\line(1,0){14}}
\put(83,0){\line(1,0){14}}
\multiput( 2.85,-1)(0,2){2}{\line(1,0){14.3}} %double line
\multiput(102.85,-1)(0,2){2}{\line(1,0){14.3}} %double line
\multiput(59,0)(4,0){6}{\line(1,0){2}} %dash line
\put(10,0){\makebox(0,0){$>$}}
\put(110,0){\makebox(0,0){$>$}}
\put(0,-5){\makebox(0,0)[t]{$0$}}
\put(20,-5){\makebox(0,0)[t]{$1$}}
\put(40,-5){\makebox(0,0)[t]{$2$}}
\put(100,-5){\makebox(0,0)[t]{$n\!\! -\!\! 1$}}
\put(120,-6.5){\makebox(0,0)[t]{$n$}}
\put(2,15){\makebox(0,0)[t]{$q^2$}}
\put(20,12){\makebox(0,0)[t]{$q$}}
\put(40,12){\makebox(0,0)[t]{$q$}}
\put(100,12){\makebox(0,0)[t]{$q$}}
\put(123,18){\makebox(0,0)[t]{$q^{\frac{1}{2}}$}}
}
\end{picture}
\end{align*}
We also let  
$\mathfrak{g}^{s,t}\,(s,t \in \{1,2\})$ denote  
the relevant affine Lie algebras as above.
$\mathfrak{g}^{2,1}={\tilde A}^{(2)}_{2n}$ is isomorphic
to $\mathfrak{g}^{1,2}=A^{(2)}_{2n}$ and their difference 
is only the enumeration of vertices.
The Langlands dual of $\mathfrak{g}^{s,t}$ is given by
${}^L\mathfrak{g}^{s,t}=\mathfrak{g}^{3-s,3-t}$.
Note that $q_0=q^{s^2/2}, \,q_n=q^{t^2/2}$ and 
$q_i=q$ for $0 < i < n$.
The coproduct $\Delta$ has the form 
\begin{align}\label{Del}
\Delta k^{\pm 1}_i = k^{\pm 1}_i\otimes k^{\pm 1}_i,\quad
\Delta e_i = 1\otimes e_i + e_i \otimes k_i,\quad
\Delta f_i = f_i\otimes 1 + k^{-1}_i\otimes f_i.
\end{align}
The opposite coproduct is denoted by $\Delta' = P \circ \Delta$,
where $P(u\otimes v) = v \otimes u$ is the exchange of the 
components.

\subsection{Homomorphism from $U_q$ to $q$-oscillator algebra}\label{ssec:emb}
Let $\A_q$ be the algebra over $\C(q^{\frac{1}{2}})$
generated by ${\bf a}^+, {\bf a}^-, {\bf k}$ and  
${\bf k}^{-1}$ obeying the relations
\begin{align}\label{qrel}
{\bf k}\,{\bf k}^{-1}={\bf k}^{-1}\,{\bf k}=1,\quad
{\bf k}\,{\bf a}^{\pm} = q^{\pm 1} {\bf a}^{\pm}\,{\bf k},
\quad
{\bf a}^{\pm}\,{\bf a}^{\mp} = 1-q^{\mp 1}{\bf k}^2.
\end{align}
The algebra $\A_q$, which we call the $q$-{\em oscillator algebra},
plays a central role in this paper.
Set
\begin{align*}
d = \frac{q}{(q-q^{-1})^2},\quad
d_1 = d\vert_{q\rightarrow q^{1/2}},\quad
d_2 = d\vert_{q\rightarrow q^2}.
\end{align*}
In what follows an element 
${\bf a}^+ \otimes 1\otimes {\bf k} \otimes {\bf a}^- \in 
\A_q^{\otimes 4}$ for example will be denoted by 
${\bf a}^+_1{\bf k}_3{\bf a}^-_4$ etc.
Thus the $q$-oscillator generators with different 
indices are commuting.
\begin{proposition}\label{pr:pi}
For a parameter $z$
the following map defines an algebra homomorphism
$\pi_z: U_q(\mathfrak{g}^{s,t})
\rightarrow \A_q^{\otimes n}[z,z^{-1}]$.
(On the left-hand side, $\pi_z(g)$ is denoted by $g$ for simplicity.)
\begin{alignat*}{3}
e_0&=z^sd_s ({\bf a}^+_1)^s, 
&\quad f_0&= z^{-s}\I^{s^2}({\bf a}^-_1)^s{\bf k}^{-s}_1,
& \quad k_0& = (\I{\bf k}_1)^s,\\
e_i&= d\,{\bf a}^-_i{\bf a}^+_{i+1}{\bf k}^{-1}_i,
&
f_i&= {\bf a}^+_i{\bf a}^-_{i+1}{\bf k}^{-1}_{i+1},
&
k_i&= {\bf k}^{-1}_i{\bf k}_{i+1}\quad(0 < i < n),\\
e_n&= \I^{t^2}d_t({\bf a}^-_n)^t{\bf k}^{-t}_n,
&
f_n&= ({\bf a}^+_n)^t,
&
k_n&= (-\I{\bf k}^{-1}_n)^t.
\end{alignat*}
\end{proposition}

The proposition can be shown by directly checking the 
relations (\ref{uqdef}).
The convention $z^{\pm s}$ rather than $z^{\pm 1}$
is just to avoid $z^{{\bf h}_3/s}$ 
in the forthcoming formula (\ref{sdef}).

\begin{remark}
If the formulas for $e_i, f_i, k_i$ with $0< i < n$ are
interpreted with $i \in \Z_n$, then Proposition \ref{pr:pi} 
gives an algebra homomorphism
$U_q(A^{(1)}_{n-1}) \rightarrow \A_q^{\otimes n}[z,z^{-1}]$.
For this case and $U_q(\mathfrak{g}^{s,t})$ without $e_0, f_0$ and $k_0$,
the homomorphism $\pi_z$ was essentially known in \cite{Ha}.
For type $A$ with $q$ roots of unity, 
similar homomorphisms and their applications have been studied in 
\cite{Ha, BKMS, DJMM, T, BB,SMS}.
\end{remark}

\section{3 dimensional $R$ and boundary vectors}\label{sec:3DR}

In Sections \ref{sec:mod} and \ref{sec:fock} 
we will consider the modular representation and the 
Fock representation of the $q$-oscillator algebra $\A_q$.
Let $M$ uniformly denote the left $\A_q$ module therein.
Then there is a unique (up to sign) involutive operator 
$R \in \mathrm{End}(M^{\otimes 3})$ \cite{BS, BMS} such that 
\begin{alignat}{2}
R\,{\bf k}_2{\bf a}^+_1 
&= ({\bf k}_3{\bf a}^+_1+{\bf k}_1{\bf a}^+_2{\bf a}^-_3)R, & \qquad
R\,{\bf k}_2{\bf a}^-_1 
&= ({\bf k}_3{\bf a}^-_1+{\bf k}_1{\bf a}^-_2{\bf a}^+_3)R,\label{1}\\
R\,{\bf a}^+_2 
&= ({\bf a}^+_1{\bf a}^+_3-{\bf k}_1{\bf k}_3{\bf a}^+_2)R, & \qquad
R\,{\bf a}^-_2 
&= ({\bf a}^-_1{\bf a}^-_3-{\bf k}_1{\bf k}_3{\bf a}^-_2)R,\label{2}\\
R\,{\bf k}_2{\bf a}^+_3 
&= ({\bf k}_1{\bf a}^+_3+{\bf k}_3{\bf a}^-_1{\bf a}^+_2)R, & \qquad
R\,{\bf k}_2{\bf a}^-_3 
&= ({\bf k}_1{\bf a}^-_3+{\bf k}_3{\bf a}^+_1{\bf a}^-_2)R,\label{3}\\
R\, {\bf k}_1 {\bf k}_2 & = {\bf k}_1 {\bf k}_2R, & \qquad
R\, {\bf k}_2 {\bf k}_3 & = {\bf k}_2 {\bf k}_3R. \label{4}
\end{alignat}
Moreover, it satisfies the tetrahedron equation (\ref{te}) in 
$\mathrm{End}(M^{\otimes 6})$.
We simply call it the 3D $R$.
It is customary to depict $R=R_{1,2,3}$ as the intersection of 
the three arrows $1, 2$ and $3$:
\[
\begin{picture}(80,45)(-40,-22)

\put(-77,-2){$R_{1,2,3} \,=$}
\put(0,0){\vector(0,1){15}}\put(0,0){\line(0,-1){15}}
%\put(-2,19){$\scriptstyle{b}$} 
\put(-2,-22){$\scriptstyle{2}$}

\put(0,0){\line(3,1){20}}\put(0,0){\vector(-3,-1){20}}
\put(22,6){$\scriptstyle{3}$} 
%\put(-28,-9){$\scriptstyle{c}$}

\put(0,0){\line(-3,1){20}}\put(0,0){\vector(3,-1){20}}
%\put(23,-9){$\scriptstyle{a}$}
\put(-26,6){$\scriptstyle{1}$}

\end{picture}
\] 
 
Let $M$ (resp. $M^\ast$) 
be a left (resp. right) $\A_q$ module.
Suppose they are equipped with the bilinear pairing 
$\langle \; | \; \rangle:  M^\ast \times M \rightarrow \C$
such that 
$\langle \tilde{m}' |m\rangle =\langle m'|\tilde{m}\rangle$ 
($\langle \tilde{m}' |:=\langle m' |g, \,
|\tilde{m}\rangle := g |m\rangle$)
for any $g \in \A_q$.

Consider the vectors $|\chi^{(s)}\rangle \in M$ and 
$\langle \chi^{(s)}| \in M^\ast$ for $s=1,2$ 
satisfying
\begin{alignat}{2}
{\bf a}^{\pm}|\chi^{(1)}\rangle &= (1 \mp q^{\mp\frac{1}{2}}{\bf k})|\chi^{(1)}\rangle,
\quad&
\langle \chi^{(1)}|{\bf a}^{\pm} &= 
\langle \chi^{(1)}|(1\pm q^{\pm\frac{1}{2}}{\bf k}),
\label{chi1}\\
{\bf a}^+|\chi^{(2)}\rangle &= {\bf a}^-|\chi^{(2)}\rangle,
\quad&
\langle \chi^{(2)}|{\bf a}^+ &= \langle \chi^{(2)}|{\bf a}^-.
\label{chi2}
\end{alignat}

\begin{proposition}$ $\cite[Prop.~4.1]{KS}\label{pr:bvec}
The following equalities in $M^{\otimes 3}$ and $M^{\ast \otimes 3}$ 
are valid for $s=1,2$:
\begin{align*}
R (|\chi^{(s)}\rangle \otimes|\chi^{(s)}\rangle \otimes|\chi^{(s)}\rangle)
&=
|\chi^{(s)}\rangle \otimes|\chi^{(s)}\rangle \otimes|\chi^{(s)}\rangle,\\
(\langle \chi^{(s)}| \otimes
\langle \chi^{(s)}| \otimes
\langle \chi^{(s)}| )R
&=\langle \chi^{(s)}| \otimes
\langle \chi^{(s)}| \otimes
\langle \chi^{(s)}|.
\end{align*}
\end{proposition}
We call these vectors {\em boundary vectors}.
In the representations considered in Sections 
\ref{sec:mod} and \ref{sec:fock},
the generator 
${\bf k}$ is expressed as $\mathrm{const}\cdot q^{\bf h}$ 
using some operator ${\bf h}$. 
It satisfies $[{\bf h}, {\bf a}^\pm]=\pm {\bf a}^\pm$ according to 
(\ref{qrel}).  
Moreover,  the relation (\ref{4}) implies 
\begin{align}\label{ice}
[R, {\bf h}_1+ {\bf h}_2]=[R, {\bf h}_2+{\bf h}_3]=0
\end{align}
for $R = R_{1,2,3}$.
It is an analog of the ice rule for the 6 vertex model \cite{Bax} and
will be referred to as the {\em conservation law}.

\section{Solution of Yang-Baxter equation}\label{sec:ybe}

\subsection{General construction}
The 3D $R$ and the boundary vectors 
enable one to construct a family of solutions of the Yang-Baxter equation
labeled with $n \ge 1$ \cite{KS, KO4}.
Consider $3n+3$ copies of $M$ labeled with 
$\alpha_1,\ldots, \alpha_n, \beta_1,\ldots, \beta_n,
\gamma_1,\ldots, \gamma_n$ and $4,5,6$.
Composing the tetrahedron equation (\ref{te})
with the spaces $1,2,3$ relabeled as
$\alpha_i, \beta_i,\gamma_i$, we get
\begin{align*}
&
x^{{\bf h}_4}(xy)^{{\bf h}_5}y^{{\bf h}_6}\bigl(R_{\alpha_1, \beta_1, 4}
R_{\alpha_1, \gamma_1, 5}
R_{\beta_1, \gamma_1, 6}\bigr)
\cdots 
\bigl(R_{\alpha_n, \beta_n, 4}
R_{\alpha_n, \gamma_n, 5}
R_{\beta_n, \gamma_n, 6}\bigr)R_{4,5,6}\\
&=
R_{4,5,6}x^{{\bf h}_4}(xy)^{{\bf h}_5}y^{{\bf h}_6}
\bigl(R_{\beta_1, \gamma_1, 6}
R_{\alpha_1, \gamma_1, 5}
R_{\alpha_1, \beta_1, 4}
\bigr)
\cdots 
\bigl(R_{\beta_n, \gamma_n, 6}
R_{\alpha_n, \gamma_n, 5}
R_{\alpha_n, \beta_n, 4}
\bigr),
\end{align*}
where we have multiplied 
$x^{{\bf h}_4+{\bf h}_5}y^{{\bf h}_5+{\bf h}_6}$ 
from the left and used (\ref{ice}) on the right-hand side.
Regard the boundary vectors in Proposition \ref{pr:bvec} 
as belonging to the spaces $4,5,6$.
Then evaluating the above relation between 
the boundary vectors one obtains the Yang-Baxter equation
\begin{align*}
S_{\boldsymbol{\alpha, \beta}}(x)
S_{\boldsymbol{\alpha, \gamma}}(xy)
S_{\boldsymbol{\beta,\gamma}}(y)
=
S_{\boldsymbol{\beta,\gamma}}(y)
S_{\boldsymbol{\alpha, \gamma}}(xy)
S_{\boldsymbol{\alpha, \beta}}(x),
\end{align*}
where $\boldsymbol{\alpha}=(\alpha_1,\ldots,\alpha_n)$ etc. 
The solution $S_{\boldsymbol{\alpha, \beta}}(z)$ takes 
a matrix product form 
with the boundary ``magnetic field" $z^{{\bf h}_3}$:
\begin{align}\label{sdef}
S_{\boldsymbol{\alpha, \beta}}(z)=\langle \chi^{(s)}|z^{{\bf h}_3}
R_{\alpha_1, \beta_1, 3}
R_{\alpha_2, \beta_2, 3}\cdots
R_{\alpha_n, \beta_n, 3}
|\chi^{(t)}\rangle.
\end{align}
The composition of the 3D $R$ and the evaluation by bra and ket vectors  
in (\ref{sdef}) are taken with respect to the space $M$ signified by $3$.
Plainly $S(z) \in \mathrm{End}(M^{\otimes n} \otimes M^{\otimes n})$ 
suppressing the dummy labels.
We will denote $S(z)$ by $S^{s,t}(z)$ when 
the dependence on $s, t \in \{1,2\}$ should be emphasized.
The formula (\ref{sdef}) is depicted as
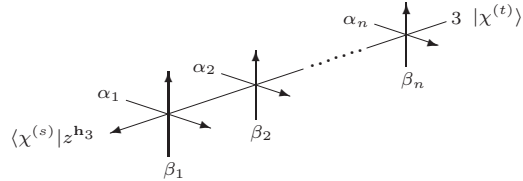
\[
\begin{picture}(200,75)(-100,-40)
\put(3,1){\vector(-3,-1){73}}

\put(-108,-27){$\scriptstyle{\langle \chi^{(s)}|z^{{\bf h}_3}}$}

%\put(-83,-27){$\scriptstyle{sc_0}$}
 
\put(-48,-16){\vector(0,1){16}}\put(-48,-16){\line(0,-1){16}}
\put(-48,-16){\vector(3,-1){16}} \put(-48,-16){\line(-3,1){16}}
%\put(-51,3){$\scriptstyle{b_1}$}
\put(-75,-10){$\scriptstyle{\alpha_1}$}
%\put(-31,-26){$\scriptstyle{a_1}$}
\put(-50,-39){$\scriptstyle{\beta_1}$}

%\put(-30,-15){$\scriptstyle{c_1}$}

\put(-15,-5){\vector(0,1){13}}\put(-15,-5){\line(0,-1){13}}
\put(-15,-5){\vector(3,-1){13}}\put(-15,-5){\line(-3,1){13}}
\put(-39,0){$\scriptstyle{\alpha_2}$}
%\put(-18,11){$\scriptstyle{b_2}$}
\put(-17,-25){$\scriptstyle{\beta_2}$}
%\put(0,-13){$\scriptstyle{a_2}$}

%\put(2,-4){$\scriptstyle{c_2}$}

\multiput(5.1,1.7)(3,1){7}{.} 
\put(6,2){
\put(21,7){\line(3,1){30}}
\put(36,12){\vector(0,1){12}}\put(36,12){\line(0,-1){12}}
\put(36,12){\vector(3,-1){12}}\put(36,12){\line(-3,1){12}}
\put(12,16){$\scriptstyle{\alpha_n}$}
%\put(50,5){$\scriptstyle{a_n}$}
%\put(33,27){$\scriptstyle{b_n}$}
\put(34,-7){$\scriptstyle{\beta_n}$}

%\put(17,1){$\scriptstyle{c_{n-1}}$}

%\put(53,17){$\scriptstyle{tc_n}$}
}

\put(59,18){$\scriptstyle{3}$} 
\put(67,19){$\scriptstyle{|\chi^{(t)}\rangle}$}
 
 \end{picture}
 \]

\subsection{Quantum group symmetry}
We supplement the $q$-oscillator algebra $\A_q^{\otimes n}$ or its representation
$\mathrm{End}(M^{\otimes n})$ with 
an invertible element $K$ satisfying the relations
\begin{align}\label{K}
K {\bf k}_j
= {\bf k}_j K,
\quad
K {\bf a}^{\pm}_j = (\I q^{\frac{1}{2}})^{\pm 1}
{\bf a}^{\pm}_j  K\quad (1 \le j \le n).
\end{align}
Introduce a slightly modified $S^{s,t}(z)$ by the so-called ``zig-zag transformation":
\begin{align}\label{Shat}
{\hat S}^{s,t}(z) &= (K \otimes 1)S^{s,t}(z)(1\otimes K^{-1}).
\end{align}
In view of ${\bf k}\,{\bf a}^{\pm} = q^{\pm 1}{\bf a}\,{\bf k}$ in (\ref{qrel})
one can formally realize $K$ 
as $K = ({\bf k}_1\ldots {\bf k}_n)^{\nu}$ for some $\nu$.
From this fact and (\ref{4}),
it follows that $[S^{s,t}(z), K\otimes K]=0$.
Using these properties one can show that
${\hat S}^{s,t}(z)$ also satisfies the 
Yang-Baxter equation.

Given parameters $x, y$ and $g \in U_q$, 
let $\Delta'(g)$ and $\Delta(g)$ simply mean the image of 
$(\pi_x \otimes \pi_y)(\Delta'(g))$ and $(\pi_x \otimes \pi_y)(\Delta(g))$
in $\mathrm{End}(M^{\otimes n}\otimes M^{\otimes n})$
by the representation of $\A_q$.
The following theorem, 
which is the main result of this section,
shows the $U_q$-symmetry of ${\hat S}^{s,t}(z)$.
The statement and the proof given in Section \ref{ssec:msk} are valid 
irrespectively of the representations of $\A_q$.
\begin{theorem}\label{th:main1}
With the choice $z=x/y$, the following commutativity holds for 
$s,t \in \{1,2\}$:
\begin{align}\label{DS}
\Delta'(g){\hat S}^{s,t}(z) 
= {\hat S}^{s,t}(z)\Delta(g)\qquad 
\forall g \in U_q(\mathfrak{g}^{s,t}).
\end{align} 
\end{theorem} 

\subsection{Proof of Theorem \ref{th:main1}}\label{ssec:msk}
It suffices to show it for $g=k_i, e_i$ and $f_i\,(0 \le i \le n)$.
The case $g= k_i$ follows easily from the first relation of 
(\ref{Del}), Proposition \ref{pr:pi} and (\ref{4}).
Let us present a proof for $g=e_i$.
In terms of $S^{s,t}(z)$ the relation 
(\ref{DS}) with $g= e_i$ takes the form
\begin{align}\label{zer}
(\tilde{e}_i\otimes 1 + k_i \otimes e_i)S^{s,t}(z) - 
S^{s,t}(z)(1\otimes \tilde{e}_i + e_i \otimes k_i)=0.
\end{align}
Here $\tilde{e}_i = K^{-1}e_i K$ and the symbol 
$\pi_x\otimes \pi_y$ is again omitted.

(i) Case $0<i<n$.
From Proposition \ref{pr:pi} and (\ref{K}) we find 
$\tilde{e}_i = e_i$.
Moreover $e_i$ and $k_i$ act nontrivially only on the factors
$R_{\alpha_i, \beta_i, 3}
R_{\alpha_{i+1}, \beta_{i+1}, 3}$ 
constituting $S^{s,t}(z)$ in (\ref{sdef}).
We rename the spaces 
$\alpha_i, \alpha_{i+1},  \beta_i, \beta_{i+1}$ as
$1, 1', 2, 2'$, respectively.
Accordingly $R_{\alpha_i, \beta_i, 3}
R_{\alpha_{i+1}, \beta_{i+1}, 3} = 
R_{1,2,3}R_{1',2',3}$ will simply be denoted by $RR'$ with the
product to be understood in the space $3$.
From Proposition \ref{pr:pi} the proof of (\ref{zer}) is reduced to showing
 \begin{align*}
 ({\bf a}^-_1{\bf a}^+_{1'}{\bf k}^{-1}_1+
 {\bf k}^{-1}_1{\bf k}_{1'}{\bf a}^-_2{\bf a}^+_{2'}{\bf k}^{-1}_2)RR'
 -RR'
 ({\bf a}^-_2{\bf a}^+_{2'}{\bf k}^{-1}_2
 +{\bf a}^-_1{\bf a}^+_{1'}{\bf k}^{-1}_1{\bf k}^{-1}_2{\bf k}_{2'})=0.
 \end{align*}
All the terms here are transformed into the form $R(\cdots)R'$ 
using the defining relations (\ref{1})--(\ref{4}) and
their alternative forms via $R=R^{-1}$ as follows.
\begin{align*}
&{\bf a}^-_1{\bf a}^+_{1'}{\bf k}^{-1}_1RR' = 
{\bf k}_2{\bf a}^-_1 R {\bf k}^{-1}_1 {\bf k}^{-1}_2 {\bf a}^+_{1'} R'
=R({\bf k}_3{\bf a}^-_1+{\bf k}_1{\bf a}^-_2{\bf a}^+_3)
{\bf k}^{-1}_1 {\bf k}^{-1}_2 {\bf a}^+_{1'} R',\\
& {\bf k}^{-1}_1{\bf k}_{1'}{\bf a}^-_2{\bf a}^+_{2'}{\bf k}^{-1}_2 RR'
= {\bf a}^-_2 R {\bf k}^{-1}_1 {\bf k}^{-1}_2{\bf k}_{1'}{\bf a}^+_{2'}R'
= R({\bf a}^-_1{\bf a}^-_3-{\bf k}_1{\bf k}_3{\bf a}^-_2)
{\bf k}^{-1}_1 {\bf k}^{-1}_2{\bf k}_{1'}{\bf a}^+_{2'}R',\\
&RR'{\bf a}^-_2{\bf a}^+_{2'}{\bf k}^{-1}_2 = 
R{\bf a}^-_2{\bf k}^{-1}_2R' {\bf a}^+_{2'} = 
R{\bf a}^-_2{\bf k}^{-1}_2({\bf a}^+_{1'}{\bf a}^+_3-{\bf k}_{1'}{\bf k}_3{\bf a}^+_{2'})R',\\
&RR' {\bf a}^-_1{\bf a}^+_{1'}{\bf k}^{-1}_1{\bf k}^{-1}_2{\bf k}_{2'} = 
R{\bf a}^-_1{\bf k}^{-1}_1{\bf k}^{-1}_2R' {\bf k}_{2'} {\bf a}^+_{1'} = 
R{\bf a}^-_1{\bf k}^{-1}_1{\bf k}^{-1}_2
({\bf k}_{3}{\bf a}^+_{1'}+{\bf k}_{1'}{\bf a}^+_{2'}{\bf a}^-_3)R'.
\end{align*}
To see the cancellation of these terms is now straightforward.

(ii) Case $i=0$ and $s=1$.
$e_0$ and $k_0$ act nontrivially only on the factor
$R=R_{\alpha_1, \beta_1, 3}$ in (\ref{sdef}).
From Proposition \ref{pr:pi} and (\ref{K}) we find 
$\tilde{e}_0 = -\I q^{-\frac{1}{2}}e_0$.
Renaming the spaces 
$\alpha_1$ and  $\beta_1$ as
$1$ and $2$, we see that (\ref{zer}) is reduced to
$0=\langle \chi^{(1)}|z^{{\bf h}_3}(-\I q^{-\frac{1}{2}}x {\bf a}^+_1+ y \I {\bf k}_1{\bf a}^+_2)R
-\langle \chi^{(1)}|z^{{\bf h}_3}R(-\I q^{-\frac{1}{2}}y{\bf a}^+_2+ x{\bf a}^+_1\I {\bf k}_2)$.
Up to an overall factor the last quantity is calculated as
\begin{align*}
&\langle \chi^{(1)}|z^{{\bf h}_3}\left(
q^{-\frac{1}{2}}z{\bf a}^+_1R
-{\bf k}_1{\bf a}^+_2R
- q^{-\frac{1}{2}}R{\bf a}^+_2
+ zR {\bf k}_2{\bf a}^+_1\right)\\
&=
\langle \chi^{(1)}|z^{{\bf h}_3}\left(
q^{-\frac{1}{2}}z{\bf a}^+_1-{\bf k}_1{\bf a}^+_2
- q^{-\frac{1}{2}}({\bf a}^+_1{\bf a}^+_3-{\bf k}_1{\bf k}_3{\bf a}^+_2)
+z({\bf k}_3{\bf a}^+_1+{\bf k}_1{\bf a}^+_2{\bf a}^-_3)\right)R.
\end{align*} 
Due to $\langle \chi^{(1)}|z^{{\bf h}_3}{\bf a}^{\pm}_3
=z^{\pm 1}
\langle \chi^{(1)}|z^{{\bf h}_3}(1\pm q^{\pm\frac{1}{2}}{\bf k}_3)$ by (\ref{chi1}),
this vanishes.

(iii) Case $i=0$ and $s=2$.
We have $\tilde{e}_0 = -q^{-1}e_0$ and (\ref{zer}) is reduced to 
\begin{align*}
0&= \langle \chi^{(2)}|z^{{\bf h}_3}\left(
q^{-1}z^2({\bf a}^+_1)^2R+{\bf k}^2_1({\bf a}^+_2)^2R
-q^{-1}R({\bf a}^+_2)^2-z^2R({\bf a}^+_1)^2{\bf k}^2_2\right)\\
&=\langle \chi^{(2)}|z^{{\bf h}_3}\left(
q^{-1}z^2({\bf a}^+_1)^2+{\bf k}^2_1({\bf a}^+_2)^2
-q^{-1}({\bf a}^+_1{\bf a}^+_3-{\bf k}_1{\bf k}_3{\bf a}^+_2)^2
-z^2({\bf k}_3{\bf a}^+_1+{\bf k}_1{\bf a}^+_2{\bf a}^-_3)^2\right)R.
\end{align*}
Due to (\ref{qrel}) and 
$\langle \chi^{(2)}|z^{{\bf h}_3}{\bf a}^+_3 
= z^2\langle \chi^{(2)}|z^{{\bf h}_3}{\bf a}^-_3$ by (\ref{chi2}),
this vanishes.
We remark that $z$ appearing in $\pi_z$ is linked to the 
$z$-commuting relation $z^{{\bf h}_3}{\bf a}_3^\pm = z^{\pm 1}
{\bf a}_3^\pm z^{{\bf h}_3}$ relevant to $i=0$, i.e., Case (ii) and Case (iii).

(iv) Case $i=n$ and $t=1$.
$e_n$ and $k_n$ act nontrivially only on the factor
$R=R_{\alpha_n, \beta_n, 3}$ in (\ref{sdef}).
From Proposition \ref{pr:pi} and (\ref{K}) we find 
$\tilde{e}_n = \I q^{\frac{1}{2}}e_n$.
Renaming the spaces 
$\alpha_n$ and  $\beta_n$ as
$1$ and $2$, we see that (\ref{zer}) is reduced to
\begin{align*}
0&=\left(q^{\frac{1}{2}}{\bf a}^-_1{\bf k}^{-1}_1-
{\bf k}^{-1}_1{\bf a}^-_2{\bf k}^{-1}_2\right)R|\chi^{(1)}\rangle
-R\left(q^{\frac{1}{2}}{\bf a}^-_2{\bf k}^{-1}_2
-{\bf a}^-_1{\bf k}^{-1}_1{\bf k}^{-1}_2\right)|\chi^{(1)}\rangle\\
&=\left(q^{\frac{1}{2}}{\bf k}_2{\bf a}^-_1-{\bf a}^-_2\right)
R{\bf k}^{-1}_1{\bf k}^{-1}_2|\chi^{(1)}\rangle
-R\left(q^{\frac{1}{2}}{\bf a}^-_2{\bf k}^{-1}_2
-{\bf a}^-_1{\bf k}^{-1}_1{\bf k}^{-1}_2\right)|\chi^{(1)}\rangle\\
&=R\left(\bigl(q^{\frac{1}{2}}({\bf k}_3{\bf a}^-_1+{\bf k}_1{\bf a}^-_2{\bf a}^+_3)
- ({\bf a}^-_1{\bf a}^-_3-{\bf k}_1{\bf k}_3{\bf a}^-_2)\bigr)
{\bf k}^{-1}_1{\bf k}^{-1}_2-
q^{\frac{1}{2}}{\bf a}^-_2{\bf k}^{-1}_2
+{\bf a}^-_1{\bf k}^{-1}_1{\bf k}^{-1}_2\right)|\chi^{(1)}\rangle.
\end{align*}
Due to ${\bf a}^{\pm}_3|\chi^{(1)}\rangle 
= (1 \mp q^{\mp\frac{1}{2}}{\bf k}_3)|\chi^{(1)}\rangle$ by (\ref{chi1}), 
this vanishes.

(v) Case $i=n$ and $t=2$.
We have $\tilde{e}_n = -qe_n$ and (\ref{zer}) is reduced to 
\begin{align*}
0&=\left(q({\bf a}^-_1)^2{\bf k}^{-2}_1+
{\bf k}^{-2}_1({\bf a}^-_2)^2{\bf k}^{-2}_2\right)R|\chi^{(2)}\rangle
-R\left(({\bf a}^-_2)^2{\bf k}^{-2}_2
+({\bf a}^-_1)^2{\bf k}^{-2}_1{\bf k}^{-2}_2
\right)|\chi^{(2)}\rangle\\
&=R\left(
\bigl(
q({\bf k}_3{\bf a}^-_1+{\bf k}_1{\bf a}^-_2{\bf a}^+_3)^2
+({\bf a}^-_1{\bf a}^-_3-{\bf k}_1{\bf k}_3{\bf a}^-_2)^2
\bigr){\bf k}^{-2}_1{\bf k}^{-2}_2
-({\bf a}^-_2)^2{\bf k}^{-2}_2
-({\bf a}^-_1)^2{\bf k}^{-2}_1{\bf k}^{-2}_2
\right)|\chi^{(2)}\rangle.
\end{align*}
Due to (\ref{qrel}) and 
${\bf a}^+_3|\chi^{(2)}\rangle = {\bf a}^-_3|\chi^{(2)}\rangle$ by (\ref{chi2}), this vanishes.
This completes the proof of (\ref{DS}) for all $g=e_i$.
The case $g=f_i$ can be verified similarly. 
\qed

\section{Example: Modular representation}\label{sec:mod}

\subsection{Modular representation of $q$-oscillator algebra}

Let $\sop,\pop$ be the generators of the Heisenberg algebra
$[\sop,\pop]=\frac{\ii}{2\pi}$.
We introduce a modular pair of the Weyl algebras, 
the exponential form of the Heisenberg algebra,  by
\begin{equation}\label{heis}
\begin{split}
&\kop=-\ii\EXP^{\pi b \sop},\,\quad \quad\wop=\EXP^{2\pi b \pop}\;,
\quad\quad \kop\wop = q\wop\kop\;,\quad q=\EXP^{\ii\pi b^2},\\
&\tilde\kop=-\ii\EXP^{\pi b^{-1} \sop}\;,\quad \tilde{\wop}=\EXP^{2\pi b^{-1} \pop}\;,\quad 
\tilde{\kop}\tilde{\wop} = \tilde{q}\tilde{\wop}\tilde{\kop}\;,\quad \tilde{q}=\EXP^{\ii\pi b^{-2}}.
\end{split}
\end{equation}
It also satisfies ${\bf k}\tilde{\bf w} = - \tilde{\bf w} {\bf k}, 
\tilde{\bf k}{\bf w} = - {\bf w}\tilde{\bf k}, 
[{\bf k}, \tilde{\bf k}]=[{\bf w}, \tilde{\bf w}]=0$.
The ``tilde'' transformation just means the replacement $b\to b^{-1}$. 
We set $\eta=\frac{1}{2}(b+b^{-1})$ 
and concentrate on the so-called strong coupling regime $0<\eta< 1$ 
in this section.
This implies that $|b|=1, \mathrm{Re}(b) = \mathrm{Re}(b^{-1}) = \eta$.

Recall that $\A_q=\langle {\bf a}^\pm, {\bf k}^\pm\rangle$ 
is the $q$-oscillator algebra (\ref{qrel}).
We call the $q$-oscillator algebra  
$\A_{\tilde q}=\langle \tilde{\bf a}^\pm, \tilde{\bf k}^\pm\rangle$
the {\em modular dual} of $\A_q$.
Identifying the generators ${\bf k}, \tilde{\bf k}$ in $\A_q, \A_{\tilde q}$ 
with those in 
the Weyl algebras, it is easy to see that
\begin{equation}\label{iden}
\begin{split}
&\bos^+=(1-q^{-1}\kop^2)^{1/2}\wop\;,\quad \bos^-=(1-q\kop^2)^{1/2}\wop^{-1},\\
&\tilde\bos^+=(1-\tilde{q}^{-1}\tilde{\kop}^2)^{1/2}\tilde{\wop}\;,\quad
\tilde\bos^-=(1-\tilde{q}\tilde{\kop}^2)^{1/2}\tilde{\wop}^{-1}
\end{split}
\end{equation}
satisfy the defining relations of $\A_q$ and $\A_{\tilde q}$.
The modular pair of the Heisenberg/Weyl algebras has the 
coordinate representations on the bra and ket vectors\footnote{
In terms of wave functions, it is a representation in the space of square integrable 
functions of $\sigma \in \R$ admitting an analytical continuation into an appropriate
horizontal strip. See \cite{Sc} for further details.} as
\begin{equation}\label{corep}
\begin{split}
&\langle\sigma|\sop=\sigma\langle\sigma|,\quad
\langle\sigma|\wop=\langle\sigma-\ii b|,\quad
\langle\sigma|\tilde{\wop}=\langle\sigma-\ii b^{-1}|,\\
&\sop |\sigma\rangle = \sigma  |\sigma\rangle,\quad
{\bf w}|\sigma\rangle = |\sigma + \I b\rangle,\quad
\tilde{\bf w}|\sigma\rangle = |\sigma + \I b^{-1}\rangle
\quad (\langle\sigma|\sigma'\rangle = \delta(\sigma-\sigma')).
\end{split}
\end{equation}
The composition of (\ref{iden}) and (\ref{corep}) will be 
referred to as {\em modular representation} of the pair
$(\A_q, \A_{\tilde q})$.
The relevant 3D $R$ is given by the integral kernel \cite{BMS}
\begin{equation}\label{kernel}
\begin{array}{l}
\ds \langle \sigma_1^{},\sigma_2^{},\sigma_3^{}|R|\sigma_1',\sigma_2',\sigma_3'\rangle = 
\delta_{\sigma_1^{}+\sigma_2^{},\sigma_1'+\sigma_2'}
\delta_{\sigma_2^{}+\sigma_3^{},\sigma_2'+\sigma_3'}
\sqrt{\frac{\varphi(\sigma_1^{})\varphi(\sigma_2^{})\varphi(\sigma_3^{})}{\varphi(\sigma_1')\varphi(\sigma_2')\varphi(\sigma_3')}}\\
[5mm]
\ds 
\times \EXP^{-\ii\pi(\sigma_1^{}\sigma_3^{}-\ii\eta(\sigma_1^{}+\sigma_3^{}-\sigma_2'))}
\int_{\mathbb{R}} du \EXP^{2\pi\ii u (\sigma_2'-\ii\eta)}
\frac{\varphi(u+\frac{\sigma_1'+\sigma_3'+\ii\eta}{2})\varphi(u+\frac{-\sigma_1^{}-\sigma_3^{}+\ii\eta}{2})}
{\varphi(u+\frac{\sigma_1^{}-\sigma_3^{}-\ii\eta}{2})\varphi(u+\frac{-\sigma_1^{}+\sigma_3^{}-\ii\eta}{2})},
\end{array}
\end{equation}
where 
$\delta_{\sigma,\sigma'}=\delta(\sigma-\sigma')$ and 
$\langle \sigma_1,\sigma_2,\sigma_3| =
\langle \sigma_1| \otimes \langle \sigma_2| \otimes \langle \sigma_3|$ etc.
The integral is convergent for $\sigma_i, \sigma'_i \in \R$ \cite{BMS}.
The function $\varphi$ is the quantum dilogarithm
\begin{equation*}
\varphi(z)=\exp \left(\frac{1}{4} \int_{\mathbb{R}+\ii 0}
\frac{\EXP^{-2\ii zw}}{\sinh(wb)\sinh(w/b)}\frac{dw}{w}\right),
\end{equation*}
which is manifestly symmetric under the exchange 
$b \leftrightarrow b^{-1}$. Its main difference property is
\begin{equation*}
\frac{\varphi(z-\ii b^{\pm 1}/2)}{\varphi(z+\ii b^{\pm 1}/2)}
= 1+\EXP^{2\pi z b^{\pm 1}}.
\end{equation*}
In fact this enables one to establish the formula (\ref{kernel}) 
by checking the defining relations (\ref{1})--(\ref{4}).

\begin{remark}\label{re:ask20}
The defining relations (\ref{1})--(\ref{4}) of $R$ are 
based on $\A_q$ but do not involve $q$ explicitly.
Moreover (\ref{kernel}) is symmetric under 
$b \leftrightarrow b^{-1}$.
It follows that the $R$ also satisfies (\ref{1})--(\ref{4})
with 
$\langle {\bf a}^\pm_i, {\bf k}^\pm_i\rangle$ replaced by
$\langle \tilde{\bf a}^\pm_i, \tilde{\bf k}^\pm_i\rangle$.
In this sense (\ref{kernel}) is the 3D $R$ for
the modular representation of $(\A_q, \A_{\tilde q})$.
This fact will further be utilized in Theorem \ref{th:main2}.
\end{remark}

\subsection{Boundary vectors}

Consider the boundary ket vectors $|\chi^{(s)}\rangle$ 
satisfying the left conditions in (\ref{chi1}) and (\ref{chi2}).
In the spirit of quantum mechanics
we denote its wave function $\langle\sigma|\chi^{(s)}\rangle$ by
$\chi^{(s)}(\sigma)$.
Then the conditions read
\begin{align}\label{x12}
\frac{\chi^{(1)}(\sigma -\ii \frac{b}{2})}{\chi^{(1)}(\sigma+\ii \frac{b}{2})} 
= \sqrt{\frac{1+\ii \EXP^{\pi b \sigma}}{1-\ii \EXP^{\pi b \sigma}}},
\qquad
\frac{\chi^{(2)}(\sigma-\ii b)}{\chi^{(2)}(\sigma+\ii b)} = 
\sqrt{\frac{1+\EXP^{2\pi b(\sigma + \ii \frac{b}{2})}}
{1+\EXP^{2\pi b(\sigma-\ii \frac{b}{2})}}}.
\end{align}
We may also set $\chi^{(s)}(\sigma) = \langle \chi^{(s)}|\sigma\rangle$ 
for the boundary bra vectors $\langle \chi^{(s)}|$ 
since $\langle \chi^{(s)}|\sigma\rangle$ 
obeys the same difference equations as (\ref{x12}). 

As we are concerned with the modular representation of $(\A_q, \A_{\tilde q})$,
it is natural to also consider the boundary vectors
$\langle \tilde{\chi}^{(s)}|$ and 
$|\tilde{\chi}^{(s)}\rangle$ obeying 
(\ref{chi1}) and (\ref{chi2}) with 
$\A_q=\langle {\bf a}^\pm, {\bf k}^\pm\rangle$ replaced by
$\A_{\tilde q}=\langle \tilde{\bf a}^\pm, \tilde{\bf k}^\pm\rangle$.
Their wave functions 
$\tilde{\chi}^{(s)}(\sigma) = \langle\sigma|\tilde{\chi}^{(s)}\rangle
= \langle \tilde{\chi}^{(s)}|\sigma \rangle$
are to satisfy (\ref{x12}) with $b$ replaced by $b^{-1}$:
\begin{align}\label{xt12}
\frac{\tilde{\chi}^{(1)}(\sigma -\ii \frac{b^{-1}}{2})}
{\tilde{\chi}^{(1)}(\sigma+\ii \frac{b^{-1}}{2})} 
= \sqrt{\frac{1+\ii \EXP^{\pi b^{-1} \sigma}}{1-\ii \EXP^{\pi b^{-1} \sigma}}},
\qquad
\frac{\tilde{\chi}^{(2)}(\sigma-\ii b^{-1})}
{\tilde{\chi}^{(2)}(\sigma+\ii b^{-1})} = 
\sqrt{\frac{1+\EXP^{2\pi b^{-1}(\sigma + \ii \frac{b^{-1}}{2})}}
{1+\EXP^{2\pi b^{-1}(\sigma-\ii \frac{b^{-1}}{2})}}}.
\end{align}
Introduce the function
\begin{equation}\label{chb}
\chi_b(\sigma) =  \exp\left(\frac{1}{8} \int_{\mathbb{R}+\ii 0}
\frac{\EXP^{-2\ii\sigma w}}{\sinh(wb)\cosh(w/b)} \frac{dw}{w}\right).
\end{equation}
It is analytic in the strip 
$-\eta < \mathrm{Im}(\sigma) < \eta$ 
but {\em not} symmetric under the exchange $b \leftrightarrow b^{-1}$.
See Appendix \ref{app:chi} for more properties.
Now we present our key observation.
\begin{proposition}\label{pr:key}
The following provides a solution to 
(\ref{x12}) and (\ref{xt12}):
\begin{equation*}
\chi^{(1)}(\sigma) = \tilde{\chi}^{(2)}(\sigma) = \chi_b(\sigma),\quad
\chi^{(2)}(\sigma) = \tilde{\chi}^{(1)}(\sigma) = \chi_{b^{-1}}(\sigma).
\end{equation*}
\end{proposition}

The statement can be verified by a direct calculation. 
For example for $\chi^{(1)}(\sigma)$, 
taking the log of the difference equation (\ref{x12}), 
making Fourier transformation assuming analyticity in the strip 
$-\frac{1}{2}\mathrm{Re}(b) 
< \mathrm{Im}(\sigma) < \frac{1}{2}\mathrm{Re}(b)$ leads to the result.
Proposition \ref{pr:key} demonstrates a curious feature of the vectors 
$|\chi^{(1)}\rangle$ and $|\chi^{(2)}\rangle$. 
They are transformed to each other simply by the interchange 
$b \leftrightarrow b^{-1}$.

\subsection{Modular double for quantum groups}

The operator ${\bf h}$ in (\ref{ice}) and its modular 
dual $\tilde{\bf h}$ satisfying 
$[{\bf h},{\bf a}^\pm]=\pm {\bf a}^\pm$ and 
$[\tilde{\bf h},\tilde{\bf a}^\pm]=\pm \tilde{\bf a}^\pm$
can be taken as 
$-ib^{-1}\boldsymbol{\sigma}$ and $-ib\boldsymbol{\sigma}$.
See (\ref{heis}) and (\ref{iden}). 
Thus in (\ref{sdef}) we may choose 
$z^{\bf h}=\EXP^{2\pi \I \lambda \boldsymbol{\sigma}}$ so that 
$\langle \sigma | z^{{\bf h}}| \sigma'\rangle 
= \delta_{\sigma, \sigma'}\EXP^{2\pi \I \lambda \sigma}$
with $\lambda$ being the (additive) spectral parameter.
Let us denote the resulting object (\ref{sdef}) by $S^{s,t}(\lambda)$.
It is given by the integral kernel
\begin{equation}\label{Sker}
\langle \boldsymbol{\alpha},\boldsymbol{\beta}|
S^{s,t}(\lambda)|\boldsymbol{\alpha}',\boldsymbol{\beta}'\rangle 
= \int \prod_{k=0}^{n} d\sigma_k 
\chi^{(s)}(\sigma_0) \EXP^{2\pi\ii \lambda \sigma_0} 
\left(\prod_{k=1}^{n} \langle \alpha_k^{},\beta_k^{},\sigma_{k-1}^{}|
R|\alpha_k',\beta_k',\sigma_{k}^{}\rangle\right) \chi^{(t)}(\sigma_{n}),
\end{equation}
where $\langle \boldsymbol{\alpha},\boldsymbol{\beta}|
= \langle \boldsymbol{\alpha}| \otimes \langle \boldsymbol{\beta}|$
with $\langle \boldsymbol{\alpha}|
= \langle \alpha_1| \otimes \cdots \otimes \langle \alpha_n|,
\langle \boldsymbol{\beta}|
= \langle \beta_1| \otimes \cdots \otimes \langle \beta_n|\,
(\alpha_i, \beta_i \in \R)$ 
and similarly for 
$|\boldsymbol{\alpha}',\boldsymbol{\beta}'\rangle$\footnote{
The symbols $\boldsymbol{\alpha}$ etc. that appeared as labels of the 
spaces in (\ref{sdef}) are used here as variables. }.
The boundary wave function $\chi^{(s)}(\sigma)$ is the one 
specified in Proposition \ref{pr:key}.
Due to the conservation law [$\delta$ factors in (\ref{kernel})],
(\ref{Sker}) is proportional to 
$\prod_{k=1}^n\delta_{\alpha_k+\beta_k,\alpha'_k+\beta'_k}$
and the integrals over $\sigma_0, \ldots, \sigma_n \in \R$ 
actually reduce to a single one.

Introduce 
${\hat S}^{s,t}(\lambda) = (K_d\otimes 1) S^{s,t}(\lambda) (1\otimes K_d^{-1})$
generalizing (\ref{Shat}).
The operator $K_d$ simultaneously 
obeying (\ref{K}) and its modular counterpart, i.e., 
$K_d \tilde{\bf k}_j =\tilde{\bf k}_j K_d,\,
K_d \tilde{\bf a}^{\pm}_j = (\I \tilde{q}^{\frac{1}{2}})^{\pm 1}
\tilde{\bf a}^{\pm}_j  K_d\; (1\! \le\! j \! \le\! n)$,
is realized as
$\langle \boldsymbol{\alpha}|K_d|\boldsymbol{\alpha}'\rangle 
= \prod_{k=1}^n \delta_{\alpha_k^{},\alpha_k'} \EXP^{\pi\eta\alpha_k}$.

Let us describe the quantum group symmetry of the 
solution ${\hat S}^{s,t}(\lambda)$ of the Yang-Baxter equation.
We prepare the representation $\rho_\lambda$ of $U_q(\mathfrak{g}^{s,t})$ 
obtained by combining 
$\pi_z:  U_q(\mathfrak{g}^{s,t})\rightarrow \A_q^{\otimes n}$ 
in Proposition \ref{pr:pi} with $z=\EXP^{-2\pi b\lambda}$ and 
the modular representation of $(\A_q, \A_{\tilde q})$.
This choice of $z$ stems from 
$z^{\bf h}{\bf a}^\pm 
= \EXP^{2\pi \I \lambda \boldsymbol{\sigma}}{\bf a}^\pm
= \EXP^{\mp 2\pi b\lambda}{\bf a}^\pm z^{\bf h}$
and the remark at the end of (iii) in the proof of Theorem \ref{th:main1}.
Similarly let $\tilde{\rho}_\lambda$ be the representation of 
$U_{\tilde q}({}^L\mathfrak{g}^{s,t})$ consisting of  
$\pi_{\tilde z}:  U_{\tilde q}({}^L\mathfrak{g}^{s,t})
\rightarrow \A_{\tilde q}^{\otimes n}$ 
with ${\tilde z}=\EXP^{-2\pi b^{-1}\lambda}$
and the modular representation of $(\A_q, \A_{\tilde q})$.
Here ${}^L\mathfrak{g}^{s,t}=\mathfrak{g}^{3-s,3-t}$ is the
Langlands dual of $\mathfrak{g}^{s,t}$ as mentioned in Section \ref{ss:qaa}.

Given parameters $\mu,\nu$, let $\Delta(g)$ and $\Delta'(g)$ 
simply mean $(\rho_\mu \otimes \rho_\nu)(\Delta(g))$ and 
$(\rho_\mu \otimes \rho_\nu)(\Delta'(g))$ for $g \in U_q(\mathfrak{g}^{s,t})$.
Similarly let they mean 
$(\tilde{\rho}_\mu \otimes \tilde{\rho}_\nu)(\Delta(g))$ and 
$(\tilde{\rho}_\mu \otimes \tilde{\rho}_\nu)(\Delta'(g))$ 
for $g \in U_{\tilde q}({}^L\mathfrak{g}^{s,t})$.
The following theorem, which is the main result in this section, 
states that the symmetry of ${\hat S}^{s,t}(\lambda)$ implied by 
Theorem \ref{th:main1} is enhanced naturally to the modular double.
\begin{theorem}\label{th:main2}
For $\mu, \nu$ such that $\lambda= \mu-\nu$, the following 
commutativity is valid:
\begin{align*}
&\Delta'(g){\hat S}^{s,t}(\lambda) = 
{\hat S}^{s,t}(\lambda)\Delta(g)\qquad 
\forall g \in U_q(\mathfrak{g}^{s,t}),\\
&\Delta'(g){\hat S}^{s,t}(\lambda) = 
{\hat S}^{s,t}(\lambda)\Delta(g)\qquad 
\forall g \in U_{\tilde q}({}^L\mathfrak{g}^{s,t}).
\end{align*}
\end{theorem}
\begin{proof}
The first relation is due to Theorem \ref{th:main1}.
The second relation is due to the first relation, 
Remark \ref{re:ask20} and Proposition \ref{pr:key}.
\end{proof}

Let $U_{q,{\tilde q}}(\mathfrak{g}^{s,t})$
denote the {\em modular double} of quantum affine algebras 
in the sense of 
$U_{q,{\tilde q}}(\mathfrak{g}_{\R})$ 
\cite[eq. (1.8)]{Ip1} with $\mathfrak{g}_{\R}$ formally 
replaced by $\mathfrak{g}^{s,t}$\footnote{
$U_q(\mathfrak{g}^{s,t})$ and 
$U_{\tilde q}({}^L\mathfrak{g}^{s,t})$ 
actually commute only 
up to sign in general. 
See e.g., \cite[Prop. 9.1]{Ip1} 
for the positive principal series representations. 
The so-called transcendental relations therein are not considered here.}.
Set $\check{S}^{s,t}(\lambda) = P\hat{S}^{s,t}(\lambda)$, 
where $P$ is defined after (\ref{Del}).
By writing the coproduct action $\Delta(U_{q,{\tilde q}})$ just as $U_{q,{\tilde q}}$,
Theorem \ref{th:main2} may be rephrased symbolically as 
\begin{alignat*}{2}
&[\check{S}^{1,1}(\lambda), \,
U_{q,{\tilde q}}(D^{(2)}_{n+1})]=0,&\quad
&[\check{S}^{2,2}(\lambda), \,
U_{q,{\tilde q}}(C^{(1)}_{n})]=0,\\
&[\check{S}^{1,2}(\lambda), \,
U_{q,{\tilde q}}(A^{(2)}_{2n})]=0, & 
&[\check{S}^{2,1}(\lambda), \,
U_{q,{\tilde q}}(\tilde{A}^{(2)}_{2n})]=0.
\end{alignat*}

\section{Example: Fock representation}\label{sec:fock}
Here we explain that the specialization of 
the results in Proposition \ref{pr:pi} -- Theorem \ref{th:main1}
to the Fock representation of the $q$-oscillator algebra $\A_q$ 
reproduces the earlier result in \cite{KO4}.
We assume $q$ is generic.
By the Fock representation of $\A_q$ (\ref{qrel}) we mean  
the following on 
$F = \oplus_{m \ge 0}\C(q^{\frac{1}{2}})|m\rangle$:
\begin{align*}
{\bf a}^+|m\rangle = |m+1\rangle,\quad
{\bf a}^-|m\rangle = (1-q^{2m})|m-1\rangle,\quad
{\bf k}^{\pm 1}|m\rangle = q^{\pm(m+\frac{1}{2})}|m\rangle.
\end{align*}
Combining this with $\pi_z$ in Proposition \ref{pr:pi} 
yields an irreducible representation of $U_q(\mathfrak{g}^{s,t})$
on $F^{\otimes n}[z,z^{-1}]$ 
except $s=t=2$.
In the latter case, $F^{\otimes n}[z,z^{-1}]$ splits into two irreducible 
$U_q(C^{(1)}_n)$ modules.
They were obtained in \cite[Prop.~1-3]{KO4} without 
factoring through $\A_q^{\otimes n}$ via $\pi_z$, and called 
``$q$-oscillator representations".

The 3D $R$ associated with the Fock representation is given by 
\begin{align*}
&R(|i\rangle \otimes |j\rangle \otimes |k\rangle) = 
\sum_{a,b,c\ge 0} R^{a,b,c}_{i,j,k}
|a\rangle \otimes |b\rangle \otimes |c\rangle,\\
&R^{a,b,c}_{i,j,k} =\delta_{a+b, i+j}\delta_{b+c, j+k}
\sum_{\lambda+\mu=b}(-1)^\lambda
q^{i(c-j)+(k+1)\lambda+\mu(\mu-k)}
\frac{(q^2)_{c+\mu}}{(q^2)_c}
\binom{i}{\mu}_{\!\!q^2}
\binom{j}{\lambda}_{\!\!q^2},
\end{align*}
where 
$\binom{m}{k}_{\!\!q}= \frac{(q)_m}{(q)_k(q)_{m-k}},\,
(q)_m = (q; q)_m$ and 
$(z;q)_m = \prod_{k=1}^m(1-z q^{k-1})$.
The sum over $\lambda, \mu$ is taken under the conditions
$\lambda+\mu=b$, $0\le \mu\le i$ and $0\le \lambda \le j$.

This solution of the tetrahedron equation was obtained  
as the intertwiner of the quantum coordinate ring 
$A_q(sl_3)$ \cite{KV}\footnote{
The formula for it on p194 in \cite{KV} contains a misprint unfortunately.}.  
It was also found 
from a quantum geometry consideration in a different gauge including 
square roots \cite{BS, BMS}. 
They were shown to be the same object 
constituting a solution of the 3D reflection equation in \cite{KO1}.
The 3D $R$  can also be identified with the transition matrix
of the PBW bases of the nilpotent subalgebra of $U_q(sl_3)$ \cite{S}. 

Let $F^\ast = \oplus_{m\ge 0}
\C(q^{\frac{1}{2}})\langle m|$ be the right $\A_q$ module defined by 
\begin{align*}
\langle m | {\bf a}^+ = (1-q^{2m})\langle m-1|,\quad
\langle m | {\bf a}^-= \langle m+1|,\quad
\langle m | {\bf k}^{\pm1}=q^{\pm(m+\frac{1}{2})} \langle m |.
\end{align*}
The bilinear pairing of $F^*$ and $F$ is specified by
$\langle m | m' \rangle= (q^2)_m\delta_{m,m'}$.
The operator ${\bf h}$ argued around (\ref{ice}) can be defined by 
$\langle m | {\bf h} = m \langle m|$ and 
${\bf h}|m\rangle = m|m\rangle$.

The boundary vectors satisfying the postulates in Section \ref{sec:3DR} 
are given by  \cite{KS}
\begin{align*}
\langle \chi^{(s)}| = \sum_{m\ge 0}\frac{1}{(q^{s^2})_m}\langle sm|,\qquad
|\chi^{(s)}\rangle = \sum_{m\ge 0}\frac{1}{(q^{s^2})_m}|sm\rangle
\quad (s=1,2).
\end{align*}

In \cite{KO4},  Theorem \ref{th:main1} was proved for the present setting of 
the Fock representation.
It was also shown that 
$F^{\otimes n}[x,x^{-1}]\otimes F^{\otimes n}[y,y^{-1}]$ with generic $x,y$ is an 
irreducible $U_q(\mathfrak{g}^{s,t})$ module 
except $s=t=2$.
(In the latter case it splits into four irreducible $U_q(C^{(1)}_n)$ modules.) 
Thus the commutativity with $U_q$ 
provides a characterization of $S^{s,t}(z)$ up to normalization.

\appendix
\section{Properties of $\chi_b(\sigma)$}\label{app:chi}
Let us collect some properties of the function
$\chi_b(\sigma)$ in (\ref{chb}) which are derived by residue analyses.
In what follows, we will also use the Jacobi modular partner 
$\overline{q}=\EXP^{-\ii\pi b^{-2}}$ to $q$.
\begin{equation*}
\frac{\chi_b(\sigma)}{\chi_b(-\sigma)}
= \EXP^{-\frac{1}{2}\pi b^{-1}\sigma},
\quad
\chi_b(\sigma)\chi_{b^{-1}}(\sigma) 
= \frac{\varphi(\frac{\sigma+\ii\eta}{2})}{\varphi(\frac{\sigma-\ii\eta}{2})}.
\end{equation*}
Set
$w=\EXP^{2\pi b \sigma},\, \overline{w}=\EXP^{2\pi b^{-1}\sigma}$.
In the regime $|q|<1$ one has the infinite product representations
\begin{align*}
&\chi_b(\sigma) = \sqrt{\frac{(-\ii q^{1/2} w^{1/2};q)_\infty}{(\ii q^{1/2} w^{1/2};q)_\infty} \frac{(-\overline{q}^3\overline{w};\overline{q}^4)_\infty}{(-\overline{q}\,\overline{w};\overline{q}^4)_\infty}},\\
&\sqrt{\varphi(\sigma)}\chi_b(\sigma)
=\frac{(-\ii q^{1/2}w^{1/2};q)_\infty}
{(-\overline{q}\,\overline{w};\overline{q}^4)_\infty},
\quad
\frac{1}{\sqrt{\varphi(\sigma)}}\chi_b(\sigma) 
= \frac{(-\overline{q}^3\overline{w};\overline{q}^4)_\infty}{(\ii q^{1/2} w^{1/2};q)_\infty}.
\end{align*}
The following integral identities hold:
\begin{align*}
&\int_\R \chi_b(\sigma)^2 \EXP^{-2\pi\ii\sigma\lambda}d\sigma 
= \EXP^{-\frac{\pi}{2} b^{-1} (2\lambda - \ii \eta)} 
\frac{\chi_b(\ii\eta-2\lambda)^2}{\cosh \pi b \lambda},\\
&\int_\R \chi_b(\sigma)\chi_{b^{-1}}(\sigma)
\EXP^{-2\pi\ii\sigma\lambda}d\sigma 
= 2\EXP^{-\ii\pi\eta^2/2} 
\frac{\varphi(2\lambda)}{\varphi(2\lambda-\ii\eta)} \EXP^{2\pi\lambda\eta}.
\end{align*}
\begin{remark}
In \cite{KS} reduction of 
the 3D $L$ operators was studied by boundary vectors
in the Fock representation leading to the
quantum $R$ matrices for the spin representations of 
$U_q(D^{(2)}_{n+1}), U_q(B^{(1)}_n)$ and $U_q(D^{(1)}_n)$.
Formulas in this appendix enable one to calculate 
the reduction in the modular representation and give exactly the same 
$R$-matrices.
\end{remark}

\section*{Acknowledgments}
This work is supported by Australian Research Council and 
Grants-in-Aid for Scientific Research No.~23340007 and No.~24540203
from JSPS.

\end{document}